\documentclass[pra,aps,reprint,a4paper,superscriptaddress,floatfix,longbibliography]{revtex4-2}
\usepackage{times}
\usepackage{graphicx}
\usepackage{epsfig}
\usepackage{amsfonts}
\usepackage{amsmath}
\usepackage{amssymb}
\usepackage{dsfont}
\usepackage{amsthm}
\usepackage{color}
\usepackage{comment}
\usepackage{braket}
\usepackage{physics}
\usepackage{multirow}
\usepackage{diagbox}
\usepackage{multirow}
\usepackage{makecell}
\usepackage{enumerate}
\usepackage{subcaption}

\usepackage{amsmath}

\newtheoremstyle{italicheader} 
  {0pt} 
  {0pt} 
  {} 
  {\parindent} 
  {\itshape} 
  {.} 
  { } 
  {\itshape #1~#2} 

\renewenvironment{proof} 
{\par\noindent\hspace{\parindent}\textit{Proof.} } 
{\qed} 

\theoremstyle{italicheader}
\newtheorem{theorem}{Theorem}
\newtheorem{proposition}{Proposition}

\newcommand{\st}[1]{{\color{blue} ST: \sl{#1}}}

\usepackage{float}
\usepackage{rotating}
\usepackage[short,c3]{optidef}
\usepackage[linesnumbered,ruled]{algorithm2e}
\usepackage{scalerel}[2016/12/29]

\usepackage[colorlinks=true,linkcolor=blue,citecolor=blue,urlcolor=blue]{hyperref}

\begin{document}
\renewcommand{\arraystretch}{1.3}


\title{Memory cost of quantum contextuality with Pauli observables}


\author{Stefan Trandafir}
\email{strandafir@us.es}
\affiliation{Departamento de F\'{\i}sica Aplicada II, Universidad de Sevilla, E-41012 Sevilla, Spain}

\author{Colm Kelleher}
\email{colm.kelleher@utbm.fr}
\affiliation{Université Marie et Louis Pasteur, UTBM, CNRS, Laboratoire Interdisciplinaire Carnot de Bourgogne ICB UMR 6303, 90010 Belfort, France}
\affiliation{Université Bourgogne Europe, CNRS, Laboratoire Interdisciplinaire Carnot de Bourgogne ICB UMR 6303, 21000 Dijon, France}

\author{Ad\'an~Cabello}
\email{adan@us.es}
\affiliation{Departamento de F\'{\i}sica Aplicada II, Universidad de Sevilla, E-41012 Sevilla,
Spain}
\affiliation{Instituto Carlos~I de F\'{\i}sica Te\'orica y Computacional, Universidad de
Sevilla, E-41012 Sevilla, Spain}


\begin{abstract}
Classically simulating the quantum contextual correlations produced by sequences of ideal measurements of compatible observables requires the measured system to have an internal memory. Computing the minimum amount of memory needed is, in general, challenging. Here, building upon the work of Kleinmann {\em et al.} [New J. Phys. 13, 113011 (2011)], we prove that the memory cost for simulating the contextuality produced by the $10$ three-qubit observables of Mermin's pentagram is only $\log_2(5) \approx 2.32$ bits, but the memory cost for simulating the contextuality produced by all $15$ two-qubit Pauli observables is, at least, $\log_2(6) \approx 2.58$ bits, thus exceeding the classical capacity of the system on which the measurements are performed. We also add results on the memory for simulating some subsets of quantum predictions.
\end{abstract}


\maketitle


\section{Introduction}


Minimally disturbing (or ideal) measurements of an observable are those that do not disturb any jointly measurable observable \cite{Cabello:2019PRA,Chiribella:2020PRR}. Simulating the quantum predictions for sequential minimally disturbing measurements of jointly measurable observables randomly chosen from a finite set requires the measured system to have an internal classical memory \cite{Cabello:2012FP,Cabello:2011UC,Kleinmann:2011NJP,Horodecki:2023}. 
Computing the minimum amount of memory needed to reproduce quantum contextual correlations is, in general, challenging. Kleinmann {\em et al.} \cite{Kleinmann:2011NJP} showed that classically simulating the quantum contextual correlations produced by sequences of ideal measurements taken from the two-qubit observables of the Peres-Mermin square \cite{Peres:1990PLA,Mermin:1990PRLb,Mermin:1993RMP} requires two bits, but simulating the quantum contextual correlations produced by the 15 two-qubit Pauli
observables requires more than two bits. This is interesting as this is the first case in which the memory needed exceeds the classical capacity of the quantum system on which the measurements are performed, which is upper bounded by the Holevo capacity \cite{Holevo:1982}. However, all we know is that two bits are not enough.


\subsection{Results}


Here, building on the work of Kleinmann {\em et al.} \cite{Kleinmann:2011NJP}, we prove the following main results:

{\em Result 1.} The memory cost for simulating the contextuality produced by the three-qubit observables of Mermin's pentagram \cite{Mermin:1990PRLb,Mermin:1993RMP} is $\log_2(5) \approx 2.32$ bits.

{\em Result 2.} The memory cost for simulating the contextuality produced by all $15$ two-qubit Pauli observables is, at least, $\log_2(6) \approx 2.58$ bits.

In addition, we prove some other results about the memory cost of partially simulating contextuality.


\subsection{Structure}


In Sec.~\ref{sec:nat}, we explain the motivation behind the study of contextuality for minimally disturbing measurements. In Sec.~\ref{sec:corr}, we describe contextuality experiments with sequential measurements. In Sec.~\ref{sec:sim}, we explain why an internal memory is needed. In Sec.~\ref{stateoftheart}, we detail the three cases of quantum state-independent contextuality that we study in this work. We also summarize what is known about the memory cost of simulating them. In Sec.~\ref{Mealy}, we introduce the tools used in \cite{Kleinmann:2011NJP} to compute the memory cost of contextuality.

In Sec.~\ref{sec:methodology}, we introduce the methods used in the proofs of our results. These proofs are presented in Sec.~\ref{sec:proofs}. Finally, in Sec.~\ref{conc} we summarize our results, connect our approach to other approaches in the literature, and list the still open problems.


\subsection{Contextuality as a generalization of nonlocality}
\label{sec:nat}


Quantum contextuality \cite{Budroni:2022RMP} is a natural generalization of Bell nonlocality \cite{Brunner:2014RMP} inspired by two observations: (i) Given a finite set of observables, the existence of nonclassical correlations require the observables to satisfy specific relations of incompatibility (i.e., lack of joint measurability). Otherwise, Vorob'ev's theorem \cite{Vorobev1959,Vorobev1962,Vorobev1967coal} prohibits them. (ii) Most of the relations of incompatibility for which Vorob'ev's theorem does not forbid nonclassical correlations {\em cannot} be realized in Bell experiments \cite{Bell:1964PHY} (because there is no way to distribute the observables among spatially separated parties respecting the relations of incompatibility). However, all of them can be realized as experiments with sequential ideal (i.e., minimally disturbing) 
measurements (i.e., measurements that do not disturb any jointly measurable observable) \cite{Cabello:2008PRL,Cabello:2019PRA}.

Vorob'ev's theorem is a result in probability theory and statistics found in the context of the study of coalition games \cite{Vorobev1959,Vorobev1967coal}. 
It addresses the following question: Given a collection of probability distributions for different subsets of variables (e.g., observables), can we find a single joint probability distribution for all the variables that is consistent with the given marginals? Vorob'ev's theorem gives necessary and sufficient conditions for such a joint distribution to exist. In the case the variables represent ideal measurements of observables, this necessary and sufficient condition is that the graph in which vertices represent the observables and edges connect jointly measurable observables lacks induced cycles of size four or larger \cite{Xu:2019PRA}. 

The quantum violation of the Clauser-Horne-Shimony-Holt Bell inequality \cite{Clauser:1969PRL} corresponds to the case of a cycle of size four (a square), in which the two observables each party can measure are not jointly measurable. The quantum violation of the Klyachko-Can-Binicio\u{g}lu-Shumovsky (KCBS) noncontextuality inequality \cite{Klyachko:2008PRL} corresponds to the case of a cycle of size five (a pentagon). In the case of a pentagon, there is no way to distribute the five observables among (any number of) spacelike separated parties while respecting the exact relations of joint measurability described by the pentagon. However, the correlation violating the KCBS can be produced in experiments in which each pair of jointly measurable observables is measured by first measuring one of them, and then the other. 
These sequential experiments have been faithfully implemented on single photons \cite{Ahrens:2013SR,Marques:2014PRL} and single ions \cite{Malinowski:2018PRA,Hu:2023NPJ}. More complex structures of joint measurability allow for more interesting forms of quantum contextuality, e.g., state-independent contextuality \cite{Cabello:2008PRL,Yu:2012PRL}, which can also be experimentally tested in sequential experiments \cite{Kirchmair:2009NAT,Amselem:2009PRL,Leupold:2018PRL}.


\subsection{Contextuality experiments} 
\label{sec:corr}


Contextuality experiments require to previously assure that the implementation satisfies the assumed relations of joint measurability and the ideality of the measurements. The sequences of measurements necessary to obtain the value of the contextuality witness (i.e., the value of the violation of the noncontextuality inequality) can only be carried out when it has been successfully verified that the setup meets the requirements (see, e.g., \cite{Leupold:2018PRL} to see how this is accomplished). However, the aim of both Bell and contextuality experiments is generating a {\em correlation}, defined as a set of probabilities of the form $p(a_1,\ldots,a_n|x_1,\ldots,x_n)$, where $x_1,\ldots,x_n$ is each one of the possible sets of jointly measurable observables (each of these sets is called a {\em context}) and $a_1,\ldots,a_n$ one of the possible combination of results.


\subsection{Resources for classically simulating contextuality}
\label{sec:sim}


Assuming that the measurement choices are independent of the hidden variables, the predictions of quantum theory for Bell experiments can be simulated by allowing a certain amount of superluminal communication  \cite{Toner:2003PRL,Pironio:2003PRA}. Similarly, the predictions of quantum theory for contextuality experiments with sequential measurements can be simulated by allowing the measured system to have a certain amount of hidden {\em memory} \cite{Cabello:2012FP,Cabello:2011UC,Kleinmann:2011NJP,Horodecki:2023}. This memory should be located in the measured system (rather than in the measuring devices). The reason is that the predictions of quantum theory hold even if we use a new measuring device every time we perform a measurement. The problem is that obtaining this minimum memory (i.e., the ``memory cost of quantum contextuality'') has been found very difficult except for simple cases.


\subsection{Simulation of state-independent contextuality produced by Pauli measurements}
\label{stateoftheart}


Here, we first present the arguably most fundamental examples of quantum state-independent contextuality for qubits. Then, we review the tools used to address the problem of the memory cost and what is known about the memory cost of simulating them.


\subsubsection{Case 1: Contextuality produced by the Peres-Mermin square}


Consider $9$ observables: $A$, $B$, $C$, $a$, $b$, $c$, $\alpha$, $\beta$, and $\gamma$. Suppose that the possible results of measuring each of them are $+1$
and $-1$. Noncontextual hidden-variable models assume that measurements reveal predetermined results that are independent of which other jointly measurable observables are measured.
For any noncontextual hidden-variable model the following inequality holds \cite{Cabello:2008PRL}: 
\begin{equation} 
    \langle ABC \rangle + \langle abc \rangle +
    \langle \alpha \beta \gamma \rangle + \langle A a \alpha \rangle + \langle B b \beta \rangle - \langle C c \gamma \rangle \le 4,
\label{nci1}
\end{equation}
where $\langle ABC \rangle$ is the mean value of the product of the results of measuring 
$A$, $B$, and $C$ (which are three jointly measurable observables) on the same system. 
Inequality \eqref{nci1} is a {\em noncontextuality inequality}. 

If we consider a two-qubit system and choose the observables as follows:
\begin{equation} \label{pms}
    \begin{bmatrix}
         A & B & C \\
        a & b & c \\
        \alpha & \beta & \gamma \\
    \end{bmatrix}
    = 
    \begin{bmatrix}
        \sigma_z \otimes \mathds{1}  & \mathds{1} \otimes \sigma_z & \sigma_z \otimes \sigma_z \\
        \mathds{1} \otimes \sigma_x & \sigma_x \otimes \mathds{1}& \sigma_x \otimes \sigma_x \\
        \sigma_z \otimes \sigma_x & \sigma_x \otimes \sigma_z & \sigma_y \otimes \sigma_y
    \end{bmatrix},
\end{equation}
then, {\em any} two-qubit state violates inequality \eqref{nci1}, since the left-hand side of \eqref{nci1} becomes $6$ [see prediction (II) in Sec.~\ref{qps}]. 
The violation predicted by quantum theory has been verified in experiments with sequential measurements on individual ions \cite{Kirchmair:2009NAT} and photons \cite{Amselem:2009PRL}.

The right-hand side of Eq.~\eqref{pms} is the so-called Peres-Mermin (or ``magic'') square \cite{Peres:1990PLA,Mermin:1990PRLb,Mermin:1993RMP}, which is a set of quantum observables used in a proof of the Kochen-Specker theorem \cite{Kochen:1967JMM}.


\subsubsection{The quantum predictions we want to simulate}
\label{qps}


Our main objective is to simulate the \emph{deterministic} predictions made by quantum theory. For example, if one measures the observable $A$ of the Peres-Mermin square twice in succession, to obtain measurement outcomes $m_1,m_2$, then it must be the case that $m_1 = m_2.$

Moreover, in this article, our objective is only to simulate a subset of the deterministic predictions made by quantum theory. In particular, we are interested in simulating the following:

(Ia) Any measurement, when repeated after measuring only observables within a single context, yields the same result. For example, if we measure $A$ followed by $B$ and then we measure $A$ again within the context $\{A,B,C\}$, and obtain, respectively, results $m_1$, $m_2$, and $m_3$, then it must occur that $m_1 = m_3$.

(Ib) Any measurement $M$, when repeated after measuring only observables that are jointly measurable with $M$, yields the same result. For example, if we measure $A$, then $B$, then $a$, and then $A$ again, and obtain respective results $m_1$, $m_2$, $m_3$, and $m_4$, then, it must occur that $m_1=m_4$, even though $B$ and $a$ are not jointly measurable.

(II) The product of the results of measuring (in any order) the three observables in each context (i.e., in each row or column of Fig.~\ref{fig:square-magicset}) must be $+1$, except for the last column, in which case the product of the results must be $-1$. This follows from the fact that the product of the self-adjoint operators that represent these observables in quantum theory is the identity (in dimension four) or minus the identity, respectively. For example, if we measure $C$ followed by $c$ and then $\gamma$, and obtain outcomes $m_1$, $m_2$, and $m_3$, respectively, then it must occur that $m_1 m_2 m_3 = -1$.

The six sets of jointly measurable observables and the indication of whether their product is the identity (in dimension four) or minus the identity is illustrated in Fig.~\ref{fig:square-magicset}.

Prediction (Ia) justifies the assumption of noncontextual hidden-variables, and prediction (Ib) represents a strengthened version [that is (Ib) implies (Ia)]. Prediction (II) leads to the violation of inequality \eqref{nci1}.

Throughout this article, we construct models to simulate predictions (Ia) and (II) [without I(b)], and also models to simulate predictions (Ia), (Ib), and (II). Our justification for these choices (as opposed to models to simulate all deterministic predictions) are as follows. Firstly, the predictions we consider cover the predictions of quantum theory associated to \emph{contextuality} (which is the phenomenon that we are interested in simulating). Secondly, the memory cost of simulating these subsets of deterministic predictions seems significantly lower than that of simulating all deterministic predictions. It is useful to understand this gap. Finally, these subsets of predictions were studied in \cite{Kleinmann:2011NJP} where exact values of the memory cost were obtained for the Peres-Mermin square. On the other hand, the exact memory cost of simulating all deterministic predictions is not known in \emph{any} case. Since very few exact results are known for the memory cost of simulating deterministic predictions, our aim was to begin to remedy this. 


\begin{figure}
\centering
\includegraphics[width=.7\linewidth]{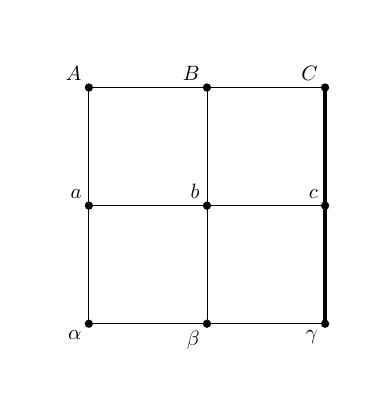}
\caption{The Peres-Mermin square. Each observable is indicated by a dot. 
Observables in the same row or column are jointly measurable. According to quantum theory, the product of the results of measuring (in any order) the three observables in each row or column is $+1$, except for the right-most column, in which case the product is $-1$. This is indicated by a thicker line.}
\label{fig:square-magicset}
\end{figure}


\subsubsection{Case 2: Contextuality produced by the 15 two-qubit Pauli observables}


Consider $15$ observables $\chi_{kl}$, where $k,l \in \{0,1,2,3\}$ and $\chi_{00}$ is excluded. Suppose that the possible results of measuring each of them are $+1$ and $-1$. For any noncontextual hidden-variable model the following inequality holds 
\cite{Cabello:2010PRA}:
\begin{equation} \label{rio}
    \begin{split}
        \sum_{k,l} \langle \chi_{k0}\chi_{kl}\chi_{0l} \rangle + \langle \chi_{11} \chi_{23} \chi_{32} \rangle + \langle \chi_{12} \chi_{21} \chi_{33} \rangle + \langle \chi_{13} \chi_{22} \chi_{31}\rangle \\ - \langle \chi_{11} \chi_{22} \chi_{33} \rangle - \langle \chi_{12} \chi_{23} \chi_{31} \rangle
- \langle \chi_{13} \chi_{21} \chi_{32}\rangle \le 9.
    \end{split}
\end{equation}

If we consider a two-qubit system and choose the observ-ables as follows:
\begin{equation} \label{epm}
    \begin{bmatrix}
         & \chi_{01} & \chi_{02} & \chi_{03} \\
        \chi_{10} & \chi_{11} & \chi_{12} & \chi_{13} \\
        \chi_{20} & \chi_{21} & \chi_{22} & \chi_{23} \\
        \chi_{30} & \chi_{31} & \chi_{32} & \chi_{33} \\
    \end{bmatrix}
    = 
    \begin{bmatrix}
        & \mathds{1} \otimes \sigma_x & \mathds{1} \otimes \sigma_y  & \mathds{1} \otimes \sigma_z \\
        \sigma_x \otimes \mathds{1} & \sigma_x \otimes \sigma_x & \sigma_x \otimes \sigma_y & \sigma_x \otimes \sigma_z \\
        \sigma_y \otimes \mathds{1} & \sigma_y \otimes \sigma_x & \sigma_y \otimes \sigma_y & \sigma_y \otimes \sigma_z \\
        \sigma_z \otimes \mathds{1} & \sigma_z \otimes \sigma_x & \sigma_z \otimes \sigma_y & \sigma_z \otimes \sigma_z \\
    \end{bmatrix},
\end{equation}
then, {\em any} two-qubit state violates inequality \eqref{rio}, since the left-hand side of \eqref{rio} becomes $15$.

The right-hand side of Eq.~\eqref{epm} is the set of all $15$ two-qubit Pauli observables (without the identity $\mathds{1} \otimes \mathds{1}$), which is sometimes called the extended Peres-Mermin set \cite{Cabello:2010PRA,Kleinmann:2011NJP}. The $15$ observables in it may be grouped into $15$ different subsets of three mutually jointly measurable observables represented by mutually commuting operators whose product is the identity (in dimension four) or minus the identity, as indicated in Fig.~\ref{fig:extendedPM-magicset}.


\begin{figure}
    \centering
\includegraphics[width=1.0\linewidth]{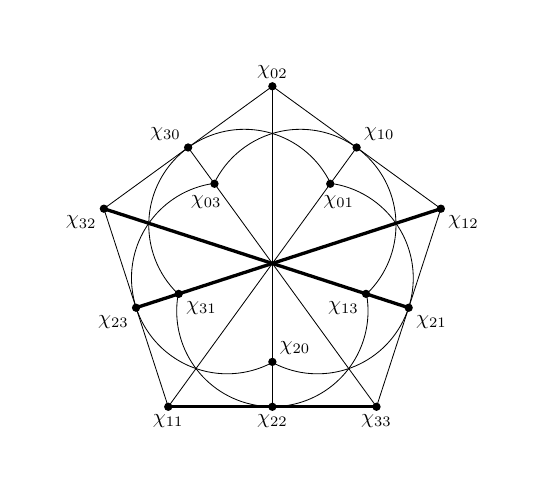}
    \caption{The extended Peres-Mermin set. Each observable is represented by a dot. Observables in the same line are jointly measurable. The product of the results of measuring the three observables in each of the three thicker lines is $-1$. For the other $12$ lines the product of the three results is $+1$. This way of presenting the observables is sometimes called the \emph{doily} \cite{PayneThas:2009}.}
    \label{fig:extendedPM-magicset}
\end{figure}


\subsubsection{Case 3: Contextuality produced by Mermin's pentagram}


Consider $10$ observables: $A$, $B$, $C$, $D$, $a_b$, $a_c$, $a_d$, $b_c$, $b_d$, and $c_d$. The possible results for measuring each of them are $+1$ and $-1$. For any noncontextual hidden-variable model, the following inequality holds \cite{Cabello:2008PRL}:
\begin{equation} \label{sta}
\begin{split}
    \langle A\, a_b\, a_c\, a_d \rangle +
    \langle a_b\, B\, b_c\, b_d \rangle +
    \langle a_c\, b_c\,C\, c_d \rangle  +
    \langle a_d\, b_d\, c_d\, D \rangle \\ -
    \langle A\, B\, C\, D \rangle \le 3.
\end{split}
\end{equation}

If we consider a three-qubit system and choose the observables as follows:
\begin{widetext}
\begin{equation} \label{mpt2}
\begin{bmatrix}
  &   & b_c &   & \\
A\; & \;B &     & C\; & \;D \\
  & a_b\; & & \;c_d & \\
    & & a_d & & \\
\;\;b_d &   & &     & a_c\;\;
     \end{bmatrix}
     =
\begin{bmatrix}
  &   & \mathds{1} \otimes \mathds{1} \otimes \sigma_x &   & \\
\sigma_x \otimes \sigma_x \otimes \sigma_z & \;\;\;\sigma_x \otimes \sigma_z \otimes \sigma_x &     & \sigma_z \otimes \sigma_x \otimes \sigma_x\;\;\; & \sigma_z \otimes \sigma_z \otimes \sigma_z \\
  & \sigma_x \otimes \mathds{1} \otimes \mathds{1}\;\;\; & & \;\;\;\sigma_z \otimes \mathds{1} \otimes \mathds{1} & \\
\mathds{1} \otimes \sigma_x \otimes \mathds{1} &   & &     & \mathds{1} \otimes \sigma_z \otimes \mathds{1}
     \end{bmatrix}
\end{equation}
\end{widetext}
then, {\em any} three-qubit state violates inequality \eqref{sta}, since the left-hand side of \eqref{sta} becomes $5$.

The right-hand side of Eq.~\eqref{mpt2} is called Mermin's (or ``magic'') pentagram (or star) \cite{Mermin:1990PRLb,Mermin:1993RMP}. The relations of joint measurability of the $10$ observables are illustrated in Fig.~\ref{fig:mermin-pentagram}.


\begin{figure}
    \centering
    \includegraphics[width=0.95\linewidth]{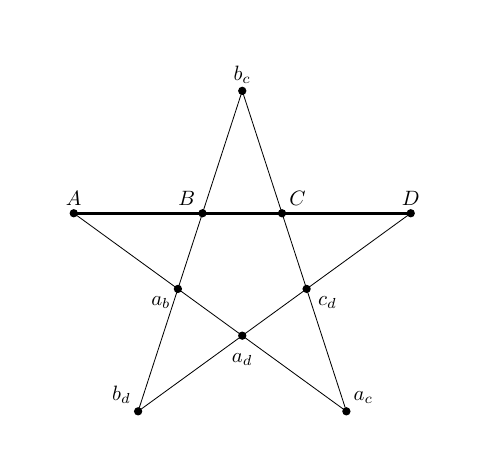}
    \caption{Mermin's  pentagram. Each observable is represented by a dot. 
    Observables in the same line are jointly measurable. The product of the results of measuring the four observables in the thicker line is $-1$. For the other four lines the product of the four results is $+1$. }
    \label{fig:mermin-pentagram}
\end{figure}


\subsection{Mealy machines}
\label{Mealy}


To classically simulate the quantum predictions for sequential measurements on a single system, Kleinmann {\em et al.} \cite{Kleinmann:2011NJP} use an automaton called a \emph{Mealy machine} \cite{Mealy:1955BSTJ}. A Mealy machine has a finite input alphabet $\mathcal{I}$, where each input corresponds to one of the observables that can be measured, a finite output alphabet $\mathcal{O}$, where the possible outputs correspond to the possible results, a set of states $\mathcal{S}$, and a function $\Omega : \mathcal{S}  \times \mathcal{I} \to \mathcal{O}$, called the \emph{output function}, which gives the output given the state and the input, and $\Upsilon : \mathcal{S}  \times \mathcal{I} \to \mathcal{S}$, called the \emph{update function}, which gives the post-measurement state given the state before the measurement and the input. 

The Mealy machine is provided with a starting state $S_1 \in \mathcal{S}$ and a sequence of inputs $(x_1,x_2,\dots)$. It then produces a sequence of corresponding outputs $(y_1,y_2,\dots)$ as follows.
It first produces output $y_1 := \Omega(S_1,x_1) \in \mathcal{O}$ corresponding to the input $x_1$. Next, the state is updated to state $S_2 := \Upsilon(S_1,x_1) \in \mathcal{S}$. This state does not need to be distinct from $S_1$. Both the output ($y_1$) and new state ($S_2$) depend only on the input ($x_1$) and the current state ($S_1$). The machine then applies the same process with $x_2$ (instead of $x_1$) and $S_2$ (instead of $S_1$) in order to produce a new output $y_2 \in \mathcal{O}$, and change to some state $S_3 \in \mathcal{S}$. In such a way, given any initial state and sequence of inputs, the machine produces a corresponding sequence of outputs, one for each input. 

This mirrors a sequence of ideal measurements on an initial quantum state. At each stage, the output depends only on the current state of the system and the observable measured (the input). Moreover, the state changes after each measurement (again only depending on the observable measured and the previous state). Importantly, while the sequence of measurement outputs in the quantum setting is probabilistic, the sequence of outputs obtained by Mealy machines is deterministic. 

In the case that the outputs produced by a Mealy machine $\mathcal{M}$ satisfy the quantum predictions of interest, we say that $\mathcal{M}$ \emph{classically simulates} those predictions. If there is no Mealy machine with fewer states satisfying the predictions, then we say that $\mathcal{M}$ is \emph{optimal} (for the set of predictions). The \emph{memory cost} of classically simulating the predictions is $\log_2(s)$, where $s$ is the number of states of an optimal Mealy machine. This mirrors the fact that there are $2^n$ distinguishable states in the quantum setting (where we view the memory cost as the number $n$ of qubits of the Pauli observables).

The following is an optimal Mealy machine $\mathcal{A}_4$ for the Peres-Mermin square:
\begin{equation} \label{eq:square-4state}
\begin{array}{ll}
S_1 = 
    \begin{bmatrix}
    + & + & (+,2) \\
    + & + & (+,3) \\
    + & + & +
\end{bmatrix},
&
S_2 =
\begin{bmatrix}
    + & + & + \\
    - & + & - \\
    (-,4) & (+,1) & +
\end{bmatrix}, \\
S_3 = 
\begin{bmatrix}
    + & - & -\\
    + & + & + \\
    (+,1) & (-,4) & +
\end{bmatrix}, &
S_4 = 
\begin{bmatrix}
    + & - & (-,3)\\
    - & + & (-,2)\\
    - & - & +
\end{bmatrix}.  \\
\end{array}
\end{equation}
It has four states: no Mealy machine with fewer states can simulate predictions (Ia), (Ib), and (II) for the Peres-Mermin square \cite{Kleinmann:2011NJP}. Each state determines the output for any of the $9$ possible measurements in the Peres-Mermin square and the post-measurement state.
If the machine is initially in state $S_1$ and the sequence $C,c,\gamma$ is measured, then the output for $C$ is $+1$ (as it is indicated by a ``$+$'' in row one, column three of $S_1$). In addition (as it is also indicated in that entry), the state changes into state $S_2$.
Next, the output for $c$ is $-1$ and the state does not change. Finally, the output of $\gamma$ is $+1$ and the state does not change. 

The previous example illustrates that the Mealy machine reproduces quantum prediction (II) for the sequence of observables $C,c,\gamma$. This does not mean that the measurements $m_1,m_2,m_3$ themselves match the outputs produced by the Mealy machine. Indeed, a cursory glance of Eq.~\eqref{eq:square-4state} confirms that $\mathcal{A}_4$ always outputs $+1$ when the input is $\gamma$ \emph{regardless} of state, whereas the value $m_3$ has no such restriction. 

We also emphasize the following fact, which is pointed out by Kleinmann \emph{et al.} \cite{Kleinmann:2011NJP}, and emphasized by Harrysson \cite{Harrysson:2016}: the four-state Mealy machine does \emph{not} simulate all of the deterministic quantum predictions. Indeed, for the sequence of inputs $A,B,c,C$, the product of the measurement outcomes of $A$, $B$, and $C$ is $1$ (since $c$ is mutually compatible with $C$). However, if the Mealy machine is initialized in state $S_1$, the corresponding outputs for the four observables $A,B,c,C$ are $+1,+1,+1,-1$, yielding a product of $(+1)(+1)(-1) = -1$ for the observables $A,B,C$. This example directly shows that the quantum predictions of (Ia), (Ib), and (II) do not cover all of the deterministic quantum predictions.

Surprisingly, not many results are known about the memory cost of simulating these quantum predictions. Indeed, only in the case of the Peres-Mermin square are there exact numbers for the memory cost in the literature, and even in that case there remain open questions. We outline the known results below.

It is shown in \cite{Kleinmann:2011NJP} that the memory cost of simulating (Ia) and (II) for the Peres-Mermin square is $\log_2(3)$ bits, and that of simulating (Ia), (Ib), and (II) is $\log_2(4) = 2$ bits [using the Mealy machine in Eq.~\eqref{eq:square-4state}].
In \cite{Kleinmann:2011NJP}, a 10-state Mealy machine is implicitly defined that simulates all deterministic predictions for the Peres-Mermin square. This machine may be viewed as a suitable $10$-element subset of the $24$ eigenvectors obtained as the four common eigenvectors of each the six contexts. The construction is made explicit in \cite{Harrysson:2016}. Optimality of this machine has not yet been proven. Therefore, the memory cost of simulating all of the deterministic predictions for the Peres-Mermin square is $\leq \log_2(10)$ bits.

In \cite{Kleinmann:2011NJP}, it is proven that, for the set of all 15 two-qubit Pauli observables, the memory cost of simulating (Ia), (Ib), and (II) is, at least, $\log_2(5)$ bits. Harrysson constructs a $27$-state Mealy machine for the set of all 15 two-qubit Pauli observables simulating all deterministic predictions of quantum theory, and so bounds the memory cost from above by $\log_2(27)$ bits. Evidently, this also bounds from above the memory cost of simulating (Ia) and (II), and (Ia), (Ib), and (II).

Finally, Harrysson \cite{Harrysson:2016} also proves an asymptotic result; that the memory cost needed in order to simulate all deterministic predictions of a set of $n$-qubit observables is bounded as $O(n^2)$. 

There are also results that do not exclusively use Mealy machines in the classical simulation, but that still contain them as necessary ingredients. In particular, by using multiple Mealy machines and an external probability distribution, one can simulate nondeterministic predictions of quantum theory. Here outputs and transitions depend not only on input and current state, but also on an additional random parameter.

Fagundes and Kleinmann \cite{Fagundes:2017JPA} showed that $\log_2(3)$ bits is sufficient to reproduce all probabilistic quantum predictions for the Peres-Mermin square for sequences within any given context [i.e., the probabilistic analogue of (Ia) and (II)]. In order to achieve this, they use $80$ modified, but equivalent, copies of the 3-state Mealy machine used to simulate deterministic predictions (Ia) and (II) for the Peres-Mermin square (along with the external source of randomness). 

It is not known what the memory cost is of simulating all probabilistic quantum predictions associated to (Ia), (Ib), and (II) of the Peres-Mermin square in this fashion.

\section{Methods} \label{sec:methodology}


We now describe the techniques that are employed in Sec.~\ref{sec:proofs}. The main tools we use to model sequences of measurements from a given Mealy machine are \emph{directed graphs}. These objects are widely used to model a variety of problems, for example in genetics, scheduling, flows in networks, and many others (see \cite{BangJensen:2008} for some of these applications). 
In this section, $b$ and $B$ do \emph{not} indicate observables of the Peres-Mermin square unless explicitly stated.


\subsection{Sets of observables and Mealy machines}


Each set of observables has an underlying \emph{incidence structure} $H = (P,B)$ that can be obtained by viewing the observables as \emph{points} (elements of $P$), and the maximal contexts as sets of observables (thus sets of points, which are elements of $B$ called \emph{blocks}). In such a way, if an observable is in a context we say that the point is \emph{incident} to the block (or vice-versa). Each of the incidence structures we study here are \emph{configurations} (each point is incident to the same number of blocks, and each block is incident to the same number of points).
It is common to study these abstract structures instead of the sets of observables themselves, since it enables one to exploit the well-known theory of configurations in order to derive properties of the sets (see, e.g., \cite{Saniga:2012}). Throughout Sec.~\ref{sec:methodology}, we also take this approach, viewing a set of observables as an incidence structure whose points are labeled with their corresponding observables.
In order to simplify notation, we say that two points $p,p'$ are \emph{compatible} if their corresponding observables are compatible.

As in Eq.~\eqref{eq:square-4state}, we view
a Mealy machine $\mathcal{M} = (\mathcal{I}, \mathcal{O}, \mathcal{S}, \Omega, \Upsilon)$ as $|\mathcal{S}|$ copies of the configuration with each state-point pair $(S_i,p_j) \in \mathcal{S} \otimes P$ labeled with its output $\Omega(S_i,p_j)$ and transition $\Upsilon(S_i,p_j)$. We refer to a state-point pair $(S_i,p_j)$ as a \emph{vertex} (for reasons which will soon become clear). We also refer to a state-block pair $(S_i,b_k)$ as a \emph{context}. For each point $p \in P$, and block $b \in B$ there are $|\mathcal{S}|$ corresponding vertices, and contexts, respectively (one for each state).

The \emph{sign} of context $(S_i,b_k)$ is defined as 
\begin{equation}
    \pi(S_i,b_k) := \prod_{p_j \in b_k} \Omega(S_i, p_{j}).
\end{equation}
If the sign of a context (state-block pair) is the same as its corresponding context, then we say that it is a \emph{confirmation context}. Otherwise we say that it is a \emph{contradiction context}. For example, for $\mathcal{A}_4$, the sign of $(S_1,\{C,c,\gamma\})$ is $+1$, while the product $Cc\gamma$ is $-I$. Therefore, $(S_1,\{C,c,\gamma\})$ is a contradiction context. On the other hand, $(S_2,\{C,c,\gamma\})$ is a confirmation context since its sign is $-1$. 

For any given state $S_i \in \mathcal{S}$, the set of outputs $\{\Omega(S_i, p) : p \in b\}$ form a $\pm 1$ assignment of the points of the underlying incidence structure. Therefore, each state must have, at least, one contradiction context. The minimum number of contradiction contexts for each state is the minimum number of unsatisfied contexts by any classical hidden variable model. This quantity, called the \emph{contextuality degree}, has been studied thoroughly in recent years (see, e.g., \cite{deBoutray:2022, Muller:2023, Muller:2024, Muller:2024XXX}). 

For the Peres-Mermin square and pentagram the contextuality degree is $1$, while the extended Peres-Mermin set has contextuality degree equal to $3$.




\subsection{Digraphs}


A \emph{directed graph} (or \emph{digraph}) $D = (V,E)$ is defined by a set of \emph{vertices} $V$ and \emph{arcs} $E \subseteq V \times V$. In our cases, the vertices of the digraph will always correspond to state-point pairs of the Mealy machine, and it is for this reason that we call state-point pairs, vertices.
We also write $V(D)$ and $E(D)$ to refer to the vertices and arcs of $D$ respectively.  
For an arc $(u,v) \in E$ we say that $v$ is an \emph{out-neighbor} of $u$, and that $u$ is an \emph{in-neighbor} of $v$. The set of out-neighbors of $v$ is often written as $N^+(v)$, and the set of in-neighbors as $N^-(v)$. This notation also applies to sets $T \subseteq V$ of vertices, so that $N^+(T)$ [$N^-(T)$] is the set of vertices $u$ in $V \setminus T$ for which there is an arc from some $v \in T$ to $u$ ($u$ to $v$).

It is common to represent digraphs pictorially by a diagram with dots (or discs) representing the vertices, and arrows (pointing from the in-neighbor to the out-neighbor) representing the arcs. See Fig.~\ref{fig:square-digraphs}(a) for such a representation.

A \emph{walk} in a graph is a sequence of vertices $v_1,\dots,v_n$ such that $v_{i+1}$ is an out-neighbor of $v_i$ for each $i=1,\dots,n-1$. The walk is \emph{closed} if $v_1 = v_n$. A \emph{path} is a walk with no repeated vertices. A \emph{cycle} is a closed walk such that removing the last vertex of the sequence leaves a path. We say that a vertex $v$ is \emph{reachable} from a vertex $u$ if there is a path from $u$ to $v$. Similarly, for sets $U,U'$ of vertices, we say that $U'$ is reachable from $U$ if there is some vertex of $U'$ that is reachable from some vertex of $U$.
A \emph{source} is a vertex with no in-neighbors, and a \emph{sink} is a vertex with no out-neighbors. A directed graph with no cycles is called a \emph{directed acyclic graph}. Every directed acyclic graph has at least one source and one sink.

A \emph{subdigraph} $D' = (V',E')$ of $D = (V,E)$ is a digraph for which $V' \subseteq V$ and $E' \subseteq E$. We say that a subdigraph is \emph{induced} if $E' = E \cap (V' \times V')$. 
A \emph{strongly connected component} of $D$ is a subdigraph $D' = (V',E')$ such that for any pair of vertices $u,v \in V(D')$, $v$ is reachable from $u$. The vertex sets of the strongly connected components of a digraph partition its vertex set.

From a given digraph $D$ and set of vertices $T \subseteq V(D)$, we may obtain another digraph by \emph{identifying} the vertices of $T$. 
This corresponds to replacing each of the vertices of $T$ with a single vertex whose set of out-neighbors  is $N^+(T)$, and whose set of in-neighbors is $N^-(T)$. By identifying each of the strongly connected components (i.e., identifying each of their vertex sets one by one) we obtain a new digraph $D'$ called the \emph{condensation} of $D$. The condensation is a directed acyclic graph and thus has sources and sinks. We call a source (sink) of $D'$ a \emph{strongly connected source (sink)} of $D$. In general, we are only interested in strongly connected sources and sinks of $D'$ and not actual sources and sinks. Therefore, when there is no possibility of confusion we simply say source or sink of $D$ in lieu of saying strongly connected source or sink.


\subsection{Mealy machine digraphs}


From a Mealy machine $\mathcal{M}$, we construct a digraph $\mathcal{D}$ by associating to every state-point pair $(S_i,p_j)$ a vertex, and by adding the arc from $(S_i,p_j)$ to $(S_k,p_{\ell})$ if and only if $\Upsilon(S_i,p_j)~=~S_k$ and $j \neq k$. Any sequence of observables from the set of observables represents a walk in the directed graph. For simplicity we elect to ignore ``loops'' (arcs from a vertex to itself) since these arcs represent repeated measurements of an observable whilst in the same state.

We call a vertex $(S_i,p_j)$ \emph{simple} if $\Upsilon(S_i,p_j) = S_i$, and \emph{nonsimple} otherwise. Similarly, we say that a context $(S_i, b)$ is \emph{simple} if each of its vertices is simple. In $\mathcal{A}_4$ the vertex $(S_2, C)$ is simple, while the vertex $(S_1,C)$ is nonsimple (since $\Upsilon(S_1,C) = S_2$. Similarly, the context $(S_2, \{C,c,\gamma\})$ is simple, while the context $(S_1, \{C,c,\gamma\})$ is not.


\subsection{Commuting digraphs}


For some subset of compatible points $R \subseteq P$, let $C(R)$ be the set of points that are compatible with all points in $R$. We define the \emph{commuting digraph} $\mathcal{D}_R$ to be the induced subdigraph of $\mathcal{D}$ with vertex set 
\begin{equation}
    \{(S_i,p) : p \in C(R), i=1,\dots,|\mathcal{S}|\}.
\end{equation}
In the case that $R = \{q\}$ consists of a single point $q \in P$, we use the short-hand $\mathcal{D}_q$ instead of $\mathcal{D}_{\{q\}}$.


\begin{figure}
    \centering
\begin{subfigure}{0.48\textwidth}
        \centering
        \includegraphics[width=\linewidth]{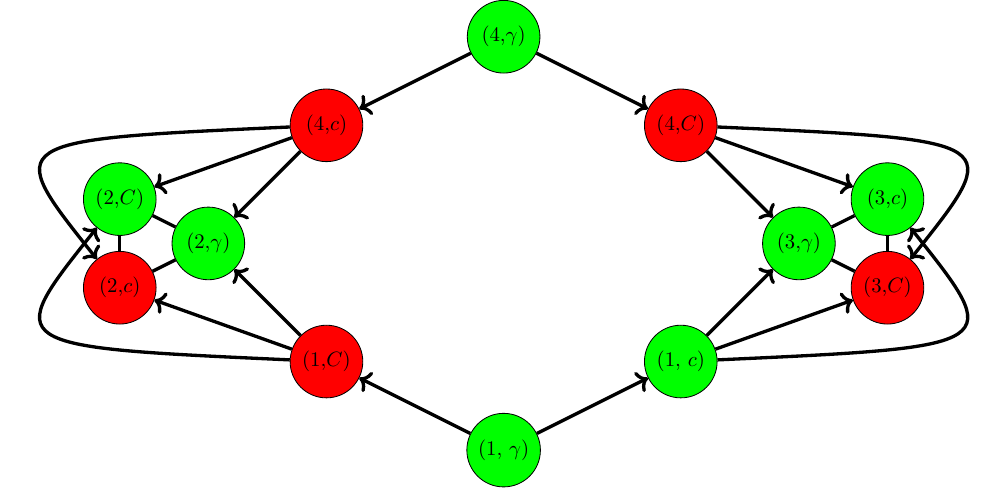}
        \caption{}
        \label{fig:square-Ccg}
    \end{subfigure}

    \vspace{.8cm}
    
    \begin{subfigure}{0.40\textwidth}
        \centering
        \includegraphics[width=\linewidth]{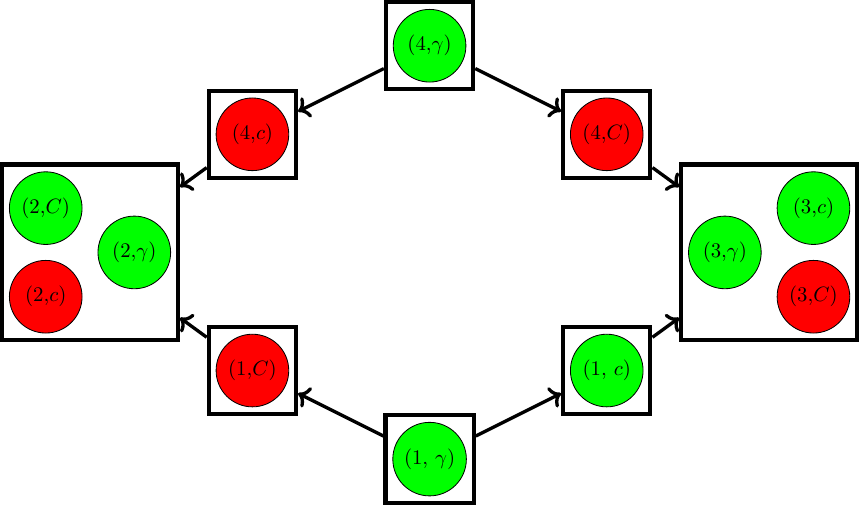}
        \caption{}
        \label{fig:square-Ccg-condensation}
    \end{subfigure}
    \caption{The digraph $\mathcal{D}_{\{C,c,\gamma\}}$ in (a) and its condensation in (b). Each vertex of (a) is labeled as $(i,p)$ where $i$ is the state number and $p$ is the input point (labeled by the corresponding observable). Each rectangle in (b) corresponds to a strongly connected component of (a) that has been identified to a single vertex, so that the digraph in (b) has $8$ vertices. Green indicates that the output is $1$, and red that the output is $-1$. }
    \label{fig:square-digraphs}
\end{figure}
\begin{figure}
    \centering
\begin{subfigure}{0.48\textwidth}
        \centering
        \includegraphics[width=\linewidth]{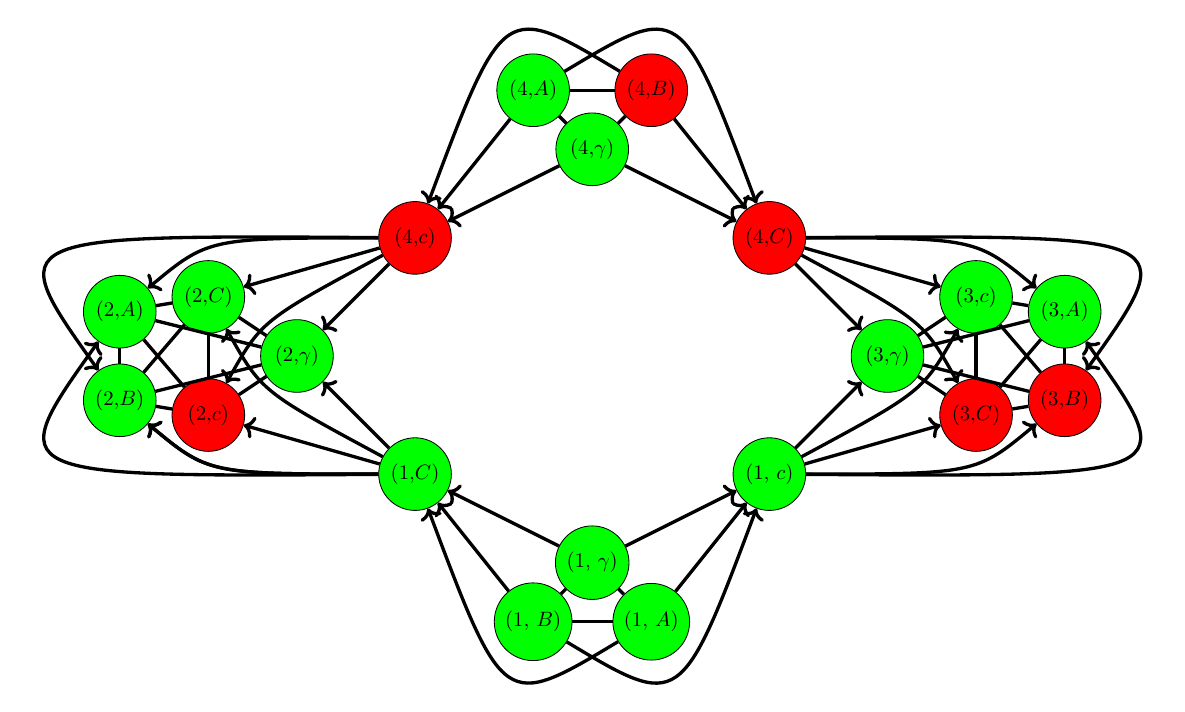}
        \caption{}
        \label{fig:square-C}
    \end{subfigure}

    \vspace{0,8cm}
    
    \begin{subfigure}{0.42\textwidth}
        \centering
        \includegraphics[width=\linewidth]{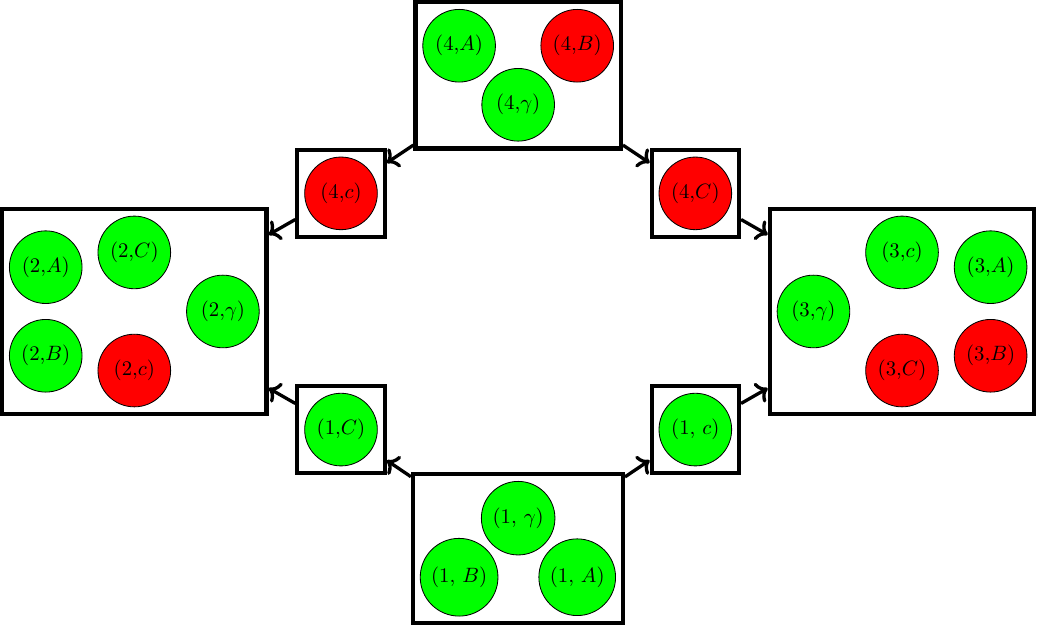}
        \caption{}
        \label{fig:square-C-condensation}
    \end{subfigure}
    \caption{The digraph $\mathcal{D}_{C}$ in (a) and its condensation in (b). Each vertex of (a) is labeled as $(i,p)$, where $i$ is the state number and $p$ is the input point (labeled by the corresponding observable). Each rectangle in (b) corresponds to a strongly connected component of (a) that has been identified to a single vertex so that the digraph in (b) has $8$ vertices.
    Green indicates that the output is $1$, and red that the output is $-1$. }
    \label{fig:square-digraphs2}
\end{figure}


Fig.~\ref{fig:square-digraphs} illustrates the digraph $\mathcal{D}_{\{C,c,\gamma\}}$ for $\mathcal{A}_4$  [in (a)] along with its condensation [in (b)].  Fig.~\ref{fig:square-C} illustrates the digraph $\mathcal{D}_{C}$ for $\mathcal{A}_4$ [in (a)] along with its condensation [in (b)]. Here (and throughout this section) we refer to the points of the Peres-Mermin configuration and their corresponding observables interchangeably. 

For each block $b \in B$, the walks of the commuting digraph $\mathcal{D}_b$ correspond to sequences of observables of the corresponding context, and for each $p \in P$, the walks of the digraph $\mathcal{D}_{\{p\}}$ correspond to sequences of observables each commuting with the observable corresponding to point $p$. 

Let us now motivate the purpose of introducing the condensations through an example with the Peres-Mermin square.  
Say that $\mathcal{A}_4$ is initialized in state $S_1$ and given the sequence of inputs $A,B,\gamma,C,\gamma$ (all compatible with $C$). This corresponds to the walk $(S_1,A), (S_1,B), (S_1,\gamma), (S_1,C), (S_2,\gamma)$ in $\mathcal{D}_C$ (see Fig.~\ref{fig:square-C} noting that vertex $(S_i, \cdot)$ is written there as $(i, \cdot)$). Notice that the arcs from $(S_1,\gamma)$ to $(S_1,C)$ and from $(S_1,C)$ to $(S_2,\gamma)$ are fundamentally different than the rest of the arcs used in the walk. Once such an arc is used, one can no longer return to the previous vertex. That is to say, after $\gamma$ is measured the first time (with the Mealy machine in state $S_1$), the Mealy machine will never take as input~$B$ in state $S_1$ again (indeed it must leave state $S_1$ entirely). This transition represents a fundamental change in the behavior of the Mealy machine. 

These changes are captured exactly by the arcs of the condensation. 
If one simulates the measurement of observables that are each compatible with $C$ at random (each with nonzero probability) using $\mathcal{A}_4$, then one necessarily ends up in a strongly connected sink, which then dictates the long-term behavior of the machine. This is true in general (for both the digraphs $\mathcal{D}_b$ and $\mathcal{D}_p$).
Unsurprisingly, understanding properties of the strongly connected sinks proves to be useful.



For a given commuting digraph $\mathcal{D}_R$, we call the vertices 
\begin{equation}
    \{(S,p) : p \in C(R)\}
\end{equation}
the \emph{entire state} $S$. Where there is no possibility of confusion, we simply refer to this as the state $S$.

\begin{proposition} \label{prop:sink-entire-state}
    Let $R \subseteq P$, and let $\mathcal{D}_R$ be the corresponding commuting digraph. The strongly connected sinks of $\mathcal{D}_R$ are unions of entire states. 
\end{proposition}

\begin{proof}
    Let $\Sigma$ be a strongly connected sink of $\mathcal{D}_R$, and let $(S_i,p_j)$ be a vertex of $\Sigma$ for some state $S_i \in \mathcal{S}$ and point $p_j \in C(R)$. If $(S_i,p_j)$ is simple, then the entire state $S_i$ is in $\Sigma$. If not, then $\Upsilon(S_i, p_j) = 
    S_k$ for some $k \neq i$. Since $\Sigma$ is strongly connected, there is a path $Q$ from the vertex $(S_k,p_j)$ to $(S_i,p_j)$. Therefore, there must be a vertex $v$ on the path $Q$ for which $\Upsilon(v) = S_i$. But then the entire state $S_i$ is a subset of the out-neighbors of $v$. This implies that the entire state is in $\Sigma$, and so $\Sigma$ must be a union of states.
\end{proof}

Since each strongly connected sink $\Sigma$ is a union of entire states, we denote $\Sigma$ simply by the set of entire states it contains. For example, the two strongly connected sinks of Fig.~\ref{fig:square-digraphs} are denoted $\{S_2\}$ and $\{S_3\}$ respectively. Moreover, we say that the \emph{order} of a strongly connected sink $\Sigma$, denoted $|\Sigma|$ is the number of entire states it contains. We write also $S \in \Sigma$ for some state $S \in \mathcal{S}$ if $\Sigma$ contains the entire state $S$.

\subsubsection{The digraphs $\mathcal{D}_b$} 
\label{sec:digraphs-Db}


In this section, we assume that $\mathcal{M} = (\mathcal{I}, \mathcal{O}, \mathcal{S}, \Omega, \Upsilon)$ is a Mealy machine satisfying predictions (Ia) and (II) for some set of observables with underlying incidence structure $H=(P,B)$, and that each commuting digraph is constructed from this Mealy machine. Moreover, we assume that $b \in B$ is a block of $H$.

\begin{proposition} \label{prop:Db-state-in-sink}
    Let $S \in \mathcal{S},\ b \in B$. If $S \in \Sigma$ for a strongly connected sink of $\mathcal{D}_b$ then $(S,b)$ is a confirmation context. If moreover, $|\Sigma| = 1$, then $(S,b)$ is a simple context.
\end{proposition}

\begin{proof}
    Since $\Sigma$ is strongly connected, there is a walk using each of the vertices of the state $S$. Therefore, predictions (Ia) and (II) imply that $(S,b)$ is a confirmation context. 
    If $|\Sigma| = 1$ (so that $\Sigma = \{S\}$), then all out-neighbors of $(S,p)$ being in $\Sigma$ must be in the state $S$.
\end{proof}

Motivated by this, we call a strongly connected sink of order $1$ a \emph{simple sink}.

The strongly connected sinks of $\mathcal{D}_b$ fully define the possible outputs obtained for the points of block $b$. We illustrate through an example. The strongly connected sinks of $\mathcal{D}_{\{C,c,\gamma\}}$ (see Fig.~\ref{fig:square-digraphs}) are $\{S_2\}, \{S_3\}$. For $S_2$ the outputs are 
\begin{equation}
\Omega(S_2, C) = 1, \Omega(S_2, c) = -1, \Omega(S_2, \gamma) = 1,
\end{equation}
and for $S_3$ the outputs are 
\begin{equation}
\Omega(S_3, C) = -1, \Omega(S_3, c) = 1, \Omega(S_3, \gamma) = 1.
\end{equation}
Therefore, for \emph{any} sequence of inputs from $\{C,c,\gamma\}$, the Mealy machine of Eq.~\eqref{eq:square-4state} must assign outputs to $C,c,\gamma$ according to either $\Omega(S_2, \cdot)$ or $\Omega(S_3, \cdot)$. For example, there is no sequence of inputs in the context $\{C,c,\gamma\}$ to the Mealy machine where $C$ and $c$ yield the same output. 

For a strongly connected sink $\Sigma$ and state $S \in \Sigma$, we use the notation $\Omega_{\Sigma}(p)$ to denote the output $\Omega(S,p)$ (valid for each state in $\Sigma$). Similarly, we denote by $\Omega_{\Sigma}(b)$ the ordered output set 
\begin{equation}
\Omega_{\Sigma}(b) := \{\Omega(S,p) : S \in \Sigma\, p \in b\}. \label{eq:sink-homogeneity}
\end{equation}

Two important tools that Kleinmann {\em et al.} \cite{Kleinmann:2011NJP} used in proving the optimality of $\mathcal{A}_4$ are related to contradiction contexts. Both of these results are stated in terms of the Peres-Mermin square, but can easily be generalized to any set of observables (as the proofs in \cite{Kleinmann:2011NJP} indicate).

The first result we state as a proposition.

\begin{proposition}\cite[Appendix C, 4]{Kleinmann:2011NJP} \label{prop:Db-2nonsimple}
    Each contradiction context $(S,b)$ must have, at least, two nonsimple vertices $(S,p),(S,p')$.
\end{proposition}

The second result (\cite[Appendix C, 4]{Kleinmann:2011NJP}) says that for each contradiction context $(S,b)$ there must be distinct states $S_1,S_2$ for which $(S_1,b),(S_2,b)$ are confirmation contexts. We generalize this.

\begin{proposition} \label{prop:Db-2sink}
    Let $(S,b)$ be a contradiction context for some $b \in B, S \in \mathcal{S}$. Then, $\mathcal{D}_b$ has, at least, two strongly connected sinks $\Sigma_1, \Sigma_2$ reachable from $S$ for which $\Omega_{\Sigma}(b) \neq \Omega_{\Sigma'}(b)$. Moreover, $S$ is in no strongly connected sink of $\mathcal{D}_b$.
\end{proposition}

\begin{proof}
    Let $\Sigma_1$ be a strongly connected sink reachable from $S$ in $\mathcal{D}_b$. Since $(S,b)$ is a contradiction context, it must be the case that $\Omega_{\Sigma_1}(p) \neq \Omega(S,p)$ 
    for some $p \in b$. Therefore, there must be a second strongly connected sink $\Sigma_2$ that is reachable from $S$ with $\Omega_{\Sigma_2}(p) = \Omega(S,p)$. Clearly, $\Omega_{\Sigma_1}(b) \neq \Omega_{\Sigma_2}(b)$.

    Since there are multiple strongly connected sinks reachable from $S$, it cannot be in a strongly connected sink.
\end{proof}

Finally, we remark that the digraphs $\mathcal{D}_b$ allow us to design a simple check for whether or not a Mealy machine satisfies predictions (Ia) and (II). Namely, we first check that each of the products of the outputs $\Omega_{\Sigma}$ of the strongly connected sinks $\Sigma$ have the correct sign. Then, we check that, for each strongly connected sink $\Sigma$, that all the vertices $v = (S,p)$ for which $\Sigma$ is reachable have the same sign as in the strongly connected sink (i.e.,  $\Omega(S,p) = \Omega_{\Sigma}(p)$). These conditions are satisfied if and only if the Mealy machine satisfies predictions (Ia) and (II). Note that the sets of vertices that can reach a strongly connected sink are called a \emph{weakly connected component} in digraph theory, and the decomposition described above is the decomposition of a digraph into its weakly connected components. 

\subsubsection{The digraphs $\mathcal{D}_p$} \label{sec:digraphs-Dp}


In Sec.~\ref{sec:digraphs-Db}, we studied the digraphs $\mathcal{D}_b$ that naturally model sequences of observables within a given context [and thus relate directly to predictions (Ia) and (II)]. In this section, we study the digraphs $\mathcal{D}_p$ corresponding to the points $p \in P$. These model sequences of observables that are all compatible with the observable corresponding to $p$ [and thus relate directly to prediction (Ib)]. We prove analogous results to the previous section, and discuss the connection between the two digraphs $\mathcal{D}_b$ and $\mathcal{D}_p$. To start, we prove results for Mealy machines that only need to satisfy predictions (Ia) and (II) [not yet assuming (Ib)]. 

Let us begin by noting the obvious: if $p \in b$ for some block $b \in B$, then $\mathcal{D}_p$ contains $\mathcal{D}_b$ as an induced subdigraph. For example, the digraph $\mathcal{D}_C$ of Fig.~\ref{fig:square-digraphs2}(a) contains the digraph
$\mathcal{D}_{\{C,c,\gamma\}}$ of Fig.~\ref{fig:square-digraphs}(a) as an induced subdigraph. This leads to the following simple, but useful observation.

\begin{proposition} \label{prop:Dp-reachability}
    Let $p \in b$, and let $S,S' \in \mathcal{S}$. If $S'$ is reachable from $S$ in $\mathcal{D}_b$, then $S'$ is reachable from $S$ in $\mathcal{D}_p$.
\end{proposition}

\begin{proof}
    Since $S'$ is reachable from $S$ in $\mathcal{D}_b$ there is some path from vertex $(S,q)$ to $(S,q')$ for points $q,q' \in b$. As $\mathcal{D}_b$ is an induced digraph of $\mathcal{D}_p$, the same path exists in $\mathcal{D}_p.$
\end{proof}

On the other hand, the digraph $\mathcal{D}_p$ contains arcs that are not in $\mathcal{D}_b$ for any block $b$ containing $p$. For example, the arc $(4,B),(4,c)$ of $\mathcal{D}_C$ is not in any digraph $\mathcal{D}_b$. It is precisely these ``jumps'' from context to context that characterize the difference between predictions (Ib) and (Ia).
As in Sec.~\ref{sec:digraphs-Db} the strongly connected sinks turn out to be important. However the extra ``jumps'' make the situation slightly more subtle.

\begin{proposition} \label{prop:Dp-state-in-sink}
    Let $S \in \mathcal{S},\
    p \in P$. Let $\Sigma$ be a strongly connected sink of $\mathcal{D}_p$ with $S \in \Sigma$.
    If $|\Sigma| \leq 2$, then $(S,b)$ is a confirmation context for each block $b \in B$ containing $p$. If $|\Sigma| = 1$, then $(S,b)$ is a simple context for each block $b \in B$ containing $p$.
\end{proposition}

\begin{proof}
    First assume that $|\Sigma| \leq 2$. Let $b \in B$ with $p \in b$. Since $S \in \Sigma$, there is at most one other state reachable from $S$ in $\mathcal{D}_p$. By Proposition~\ref{prop:Dp-reachability}, this is also the case in $\mathcal{D}_b$. Therefore, by Proposition~\ref{prop:Db-2sink}, $(S,b)$ cannot be a contradiction context, and is thus a confirmation context.

    Second assume that $\Sigma = \{S\}$ so that $|\Sigma| = 1$. Then, no other state is reachable from $S$ in $\mathcal{D}_p$ (and thus $\mathcal{D}_b$). Therefore, $(S,b)$ is simple.
\end{proof}

As in the case of $\mathcal{D}_b$, we will refer to a strongly connected sink of order $1$ as a simple sink.

For example, the digraph $\mathcal{D}_C$ (Fig.~\ref{fig:square-digraphs2}) has two strongly connected sinks: $\{S_2\},\{S_3\}$. Therefore, the two blocks containing $C$, $b_1 = \{A,B,C\}, b_2 = \{C,c,\gamma\}$, are both simple for both of the states $S_2,S_3$ (i.e., $(S_2,b_1), (S_2,b_2),(S_3,b_1), (S_3,b_2)$ are all simple contexts).

We now prove results for Mealy machines satisfying each of the predictions (Ia), (Ib), and (II). That is,  $\mathcal{M} = (\mathcal{I}, \mathcal{O}, \mathcal{S}, \Omega, \Upsilon)$ now represents a Mealy machine satisfying predictions (Ia), (Ib), and (II) for some set of observables with underlying incidence structure $H=(P,B)$.

\begin{proposition} \label{prop:sink-homogeneous-p}
    Let $p \in P$, and let $\Sigma$ be a strongly connected sink of $\mathcal{D}_p$. For each pair of states $S,S' \in \Sigma$, $\Omega(S,p) = \Omega(S',p)$.
\end{proposition}

\begin{proof}
    Since $\Sigma$ is a strongly connected component, there is a path from $(S,p)$ to $(S',p)$. Therefore, since the Mealy machine satisfies prediction (Ib), $\Omega(S,p) = \Omega(S',p)$.
\end{proof}

Note that unlike in the $\mathcal{D}_b$ case, this only holds for the point $p$, so that strongly connected sinks do \emph{not} define the outputs of all of the points.

If, for some point $p \in P$, the digraph $\mathcal{D}_p$ has at least two strongly connected sinks, we say that $p$ is a \emph{multi-sink} point.

\begin{proposition} \label{prop:Dp-2sink}
    Let $(S,b)$ be a contradiction context for some $b \in B, S \in \mathcal{S}$. Then, there are, at least, two multi-sink points $p,p' \in b$ such that $S$ is not in a sink of $\mathcal{D}_p$ nor $\mathcal{D}_{p'}$.  
\end{proposition}

\begin{proof}
    Let $r \in b$, and assume towards a contradiction that for each $q \in b \setminus \{r\}$ that either $q$ is not a multi-sink point, or $S$ is in a sink of $\mathcal{D}_q$. In both of these scenarios, there is a unique sink $\Sigma$ reachable from $S$ in $\mathcal{D}_q$. In particular, $\Omega(S',q)$ must be the same for any vertex $(S',q)$ reachable from $S$. 

    Since $\mathcal{D}_b$ is an induced subdigraph of each of these digraphs, it follows that in each sink $\Sigma'$ reachable from $S$ in $\mathcal{D}_b$ we have $\Omega_{\Sigma'}(q) = \Omega(S,q)$ for each $q \in b \setminus \{r\}$. By Proposition~\ref{prop:Db-2sink}, there must be at least two such sinks $\Sigma_1,\Sigma_2$ of $\mathcal{D}_b$ with $\Omega_{\Sigma_1}(b) \neq \Omega_{\Sigma_2}(b)$. But since both of these sinks has the same product, they must differ on at least two points. This is a contradiction.
\end{proof}


\section{Proofs} \label{sec:proofs}


\subsection{Memory cost of satisfying predictions (Ia) and (II) for Mermin's pentagram}


As illustrated in Fig.~\ref{fig:mermin-pentagram}, the underlying configuration is a set of five contexts, such that each pair intersects at exactly one point. Therefore, we may label the contexts via $b_1,\dots,b_5$, and the points by the pairs $\{(i,j) : 1 \leq i < j \leq 5\}$ so that each point $(i,j)$ lies on contexts $b_i$ and $b_j$, and each context $b_k$ consists of exactly the points $(i,j)$ for which $k \in \{i,j\}$. We write $ij$ as a short-hand for $(i,j)$.



To simplify notation, we write the Mealy machines for the pentagram using two-row arrays where the entries correspond to the observables as follows:
\begin{equation}
\begin{bmatrix}
    A & B & b_c & C & D \\
    b_d & a_b & a_d & c_d & a_c
\end{bmatrix}.
\end{equation}.

The following is a four-state Mealy machine that satisfies predictions (Ia) and (II): 

\begin{equation} \label{eq:pentagram-4state}
\scalebox{0.8}{$
    \begin{array}{ll}
    S_1 = 
    \begin{bmatrix}
      + & + & + & (+,2) & (+,3) \\
      + & + & + & + & + 
    \end{bmatrix},
    &
    S_2 = 
    \begin{bmatrix}
      + & + & + & + & - \\
      (-,4) & (+,1) & + & + & +
    \end{bmatrix}, \\ \\
    S_3 = 
    \begin{bmatrix}
        + & + & + & - & + \\
        (+,1) & (-,4) & + & + & -
    \end{bmatrix},
    &
    S_4 = 
    \begin{bmatrix}
        + & + & + & (-,3) & (-,2) \\
        - & - & + & + & -
    \end{bmatrix}.
    \end{array}
    $}%
\end{equation}

\begin{proposition}
    There is no three-state Mealy machine for Mermin's pentagram that satisfies predictions (Ia) and (II).
\end{proposition}

\begin{proof}
    Assume towards a contradiction that there does exist such a machine, and let $S_1,S_2,S_3$ denote the the three states. 

The first state $S_1$ has some contradiction context $(S_1,b_1)$. By Proposition~\ref{prop:Db-2sink}, $\mathcal{D}_{b_1}$ has two strongly connected sinks neither of which contains $S_1$. Therefore, these sinks must be exactly $\{S_2\},\{S_3\}$. By Proposition~\ref{prop:Db-state-in-sink}, $(S_2,b_1)$ and $(S_3,b_1)$ are simple contexts, and therefore confirmation contexts.

Making the same argument for $S_2$ and $S_3$ in lieu of $S_1$, we find that $(S_1,b_1), (S_2, b_2), (S_3,b_3)$ must be contradiction contexts (without loss of generality). 

For each $i$, the contexts $\{(S_i,b_j) : 1 \leq j \leq 3, j \neq i\}$ are simple. This constitutes $7$ simple vertices for each state. For each state $S_i$ the remaining three vertices are $i4, i5, 45$.
Since each state $S_i$ has at least two nonsimple vertices (by Proposition~\ref{prop:Db-2nonsimple}), there is some nonsimple vertex $(S_i,p)$ with $p \in b_4$ and some nonsimple vertex $(S_i,p')$ with $p' \in b_5$ for each $i=1,2,3$.
 
The block $b_4$ contains nonsimple vertices for each of the states $S_1,S_2,S_3$, and so none of these states can be in a strongly connected sink by itself. Therefore, $\mathcal{D}_{b_4}$ has a single strongly connected sink $\Sigma$, and $\Omega(S_i,p) = \Omega_{\Sigma}(p)$ for each $p \in b_4$ and $i=1,2,3$. The same holds true for $b_5$.  

Examining $(S_1,b_1)$, we find that the vertices $(S_1,12),(S_1,13)$ are simple (since $b_2,b_3$ are simple contexts), and that the points $14,15$ have the same output for each state. This is a contradiction since the sequence of measurements $(S_1,12), (S_1,13), (S_1,14), (S_j,15)$ [where $S_j = \Upsilon(S_1, 15)$] will have the same product of outputs as $(S_1,12), (S_1,13), (S_1,14), (S_1,15)$. But since this is a contradiction context, this means that the Mealy machine does not satisfy prediction (II).
\end{proof}


\subsection{Memory cost of satisfying predictions (Ia), (Ib), and (II) for Mermin's pentagram} \label{sec:pentagram-iabii}



The following is a five-state Mealy machine satisfying predictions (Ia), (Ib), and (II) for Mermin's pentagram: 

\begin{equation} \label{eq:pentagram-5state}
\scalebox{0.8}{$
\begin{array}{ll}
S_1 = 
\begin{bmatrix}
    (+,4) & + & + & + & (+,3) \\
    + & + & + & + & +
\end{bmatrix}, &
S_2 = 
\begin{bmatrix}
    (+,4) & + & - & + & - \\
    - & + & + & + & (-,5)
\end{bmatrix}, \\ \\
S_3 = 
\begin{bmatrix}
    - & + & (+,1) & + & + \\
    + & + & + & + & (-,5)
\end{bmatrix}, 
&
S_4 = 
\begin{bmatrix}
    + & + & (+,1) & + & - \\ 
    (-,2) & + & + & + & +
\end{bmatrix}, \\ \\
S_5 = 
\begin{bmatrix}
    - & + & - & + & (+,3) \\
    (-,2) & + & + & + & -
\end{bmatrix}.
\end{array} 
$}%
\end{equation}

The remainder of this section is dedicated to proving that this automaton is optimal, thus showing that the memory cost is $\log_2(5)$ bits. Throughout, we assume towards a contradiction that
$\mathcal{M}$ is a four-state Mealy machine for Mermin's pentagram that satisfies predictions (Ia), (Ib), and (II). 

\begin{proposition} \label{prop:star-one-sink}
    Let $p,p'$ be distinct points. No state can be in a simple sink in both $\mathcal{D}_p$ and $\mathcal{D}_{p'}$.
\end{proposition}

\begin{proof}
    Assume towards a contradiction that a state $S$ is in a simple sink in both $\mathcal{D}_p$ and $\mathcal{D}_{p'}$.
    The points $p$ and $p'$ are incident to, at least, three blocks (say $b_1,b_2,b_3$ without loss of generality) of Mermin's pentagram, which in turn are incident to nine of its ten points. 
    By Proposition~\ref{prop:Dp-state-in-sink},
    each of the contexts $(S,b_1),(S,b_2),(S,b_3)$ is simple. Therefore, there is at most one nonsimple vertex $(S,r)$ for state $S$.
    This is a contradiction since each state must have, at least, two nonsimple vertices. 
\end{proof}

Let $X_i$ denote the set of multi-sink points for which there are exactly $i$ states that are not in any strongly connected sink, and let $x_i = |X_i|$. Since there are at least two strongly connected sinks in $\mathcal{D}_p$ for each multi-sink point $p$, there must be at least two states that are in strongly connected sinks, and so we find that $i \in \{0,1,2\}$.
Note that for $i \neq j$, the sets $X_i$ and $X_j$ must be disjoint.

By Proposition~\ref{prop:Dp-2sink}, each one of the four states $S$ must have, at least, two nonsimple points, and thus there must be, at least, two points for which $S$ is no strongly connected sink. 

Therefore,
\begin{equation} \label{ineq:star-lower}
    x_1 + 2x_2 \geq 4\cdot 2 = 8.
\end{equation}
For each multi-sink point $p \in X_1$, there must be at least one state in a simple sink of $\mathcal{D}_p$ (since there are three remaining states for two strongly connected sinks). By the same token, for each multi-sink point $p' \in X_2$ there must be exactly two simple sinks in $\mathcal{D}_{p'}$. Therefore, by Proposition~\ref{prop:star-one-sink}, it follows that
\begin{equation} \label{ineq:star-upper}
    x_1 + 2x_2 \leq 4.
\end{equation}
Clearly, one cannot satisfy both inequality~\eqref{ineq:star-lower} and inequality~\eqref{ineq:star-upper}. Therefore, it is impossible to construct a four-state Mealy machine satisfying predictions (Ia), (Ib), and (II) for Mermin's pentagram.

\subsection{Contradiction contexts of the set of all two-qubit observables} \label{sec:doily-contras}


The set of all two-qubit observables presents one significant challenge in comparison to the Peres-Mermin square and Mermin's pentagram. For both of those two sets it is possible to assign a classical hidden variable model that fails to produce the proper product for a single block. In the case of the set of all two-qubit observables, there must be at least three blocks which fail, and so there are at least three contradiction contexts per state. Moreover, it is not sufficient to consider only the case with three contradiction contexts, one needs to also consider cases with four and five contradiction contexts. This is because the four and five context cases do not contain the three context case (in the sense that the three contexts do not occur as a subset of the four or the five). Nor does the five context case contain the four context case. 

Up to symmetry there are three minimal cases to consider, each illustrated in Fig.~\ref{fig:doily-contras}. We can bound the number of nonsimple vertices needed for each state, by exploiting the fact that each contradiction context needs two nonsimple vertices (Proposition~\ref{prop:Db-2nonsimple}).

Surprisingly, while the three and four contradiction context cases (illustrated in (a) and (b) of Fig.~\ref{fig:doily-contras}, respectively) need six nonsimple vertices, the five contradiction context case [illustrated in (c)] needs only five nonsimple vertices. That is, one can choose five nonsimple vertices so that each contradiction context has two nonsimple vertices.

The following proposition immediately follows by applying Propositions~\ref{prop:Db-2nonsimple} and \ref{prop:Dp-2sink}.

\begin{proposition} \label{prop:doily-nonsimple}
    Each state $S$ of a Mealy machine $\mathcal{M}$ satisfying predictions (Ia) and (II) for the set of all two-qubit observables must have, at least, five nonsimple vertices. Moreover, if $\mathcal{M}$ satisfies prediction (Ib) as well, then there must be, at least, five multi-sink points for which $S$ is not in any strongly connected sink.
\end{proposition}


\begin{figure}
    \centering
    \begin{subfigure}{0.23\textwidth}
        \centering
        \includegraphics[width=\linewidth]{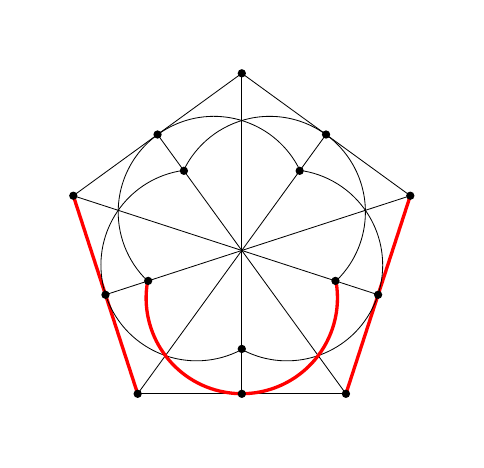}
        \caption{}
        \label{fig:doily-contras1}
    \end{subfigure}
    \begin{subfigure}{0.23\textwidth}
        \centering
        \includegraphics[width=\linewidth]{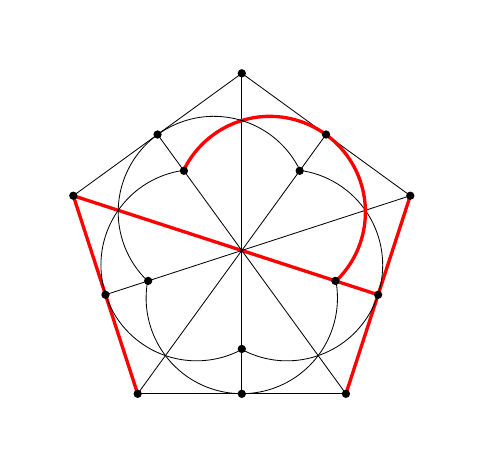}
        \caption{}
        \label{fig:doily-contras2}
    \end{subfigure}
    \begin{subfigure}{0.23\textwidth}
        \centering
        \includegraphics[width=\linewidth]{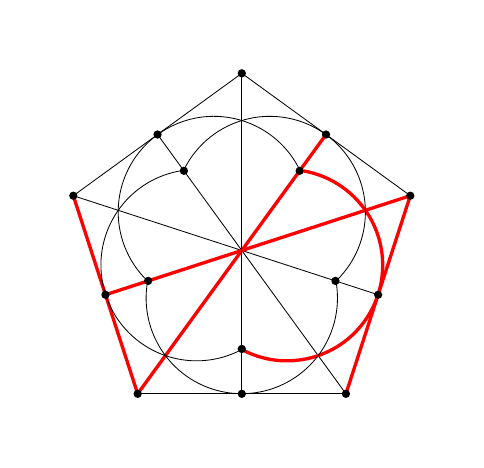}
        \caption{}
        \label{fig:doily-contras3}
    \end{subfigure}
    \caption{The minimal sets of contradiction contexts of the set of all two-qubit observables, up to symmetry (shown in red). (a) illustrates the unique set with three contradiction contexts, (b) with four, and (c) with five.}
    \label{fig:doily-contras}
\end{figure}


 It is well-known (see, e.g., \cite{Godsil:2013}) that the contexts of the set of all two-qubit observables may be labeled with pairs $\{i,j\}$ of distinct numbers from $\{1,2,\dots,6\}$, so that two contexts $\{i,j\}, \{k,\ell\}$ meet at a point if and only if $\{i,j\}$ and $\{k,\ell\}$ are disjoint sets. In such a way, each point of the set of all two-qubit observables is labeled by a partition of $\{1,2,\dots,6\}$ into three subsets of size two.

The contradiction contexts of Fig.~\ref{fig:doily-contras}(a) may then be described as simply choosing three distinct numbers $i,j,k$ from $\{1,2,\dots,6\}$ and then forming the contexts $\{i,j\},\{j,k\},\{i,k\}$. One such example is $\{1,2\},\{2,3\},\{1,3\}$.
There are $\binom{6}{3} = 20$ such choices.

The contradiction contexts of Fig.~\ref{fig:doily-contras}(b) can be described as choosing a pair $\{i,j\}$ of distinct numbers from $\{1,2,\dots,6\}$, along with a third number $k \in \{1,2,\dots,6\}$ different from $i$ or $j$. One then takes the context $\{i,j\}$ along with the  three contexts $\{k,\ell_1\}, \{k,\ell_2\}, \{k,\ell_3\}$, where $\ell_1,\ell_2,\ell_3$ are the other three numbers in $\{1,2,\dots,6\}$ distinct from $i,j,k$. One such example is $\{1,2\}, \{3,4\}, \{3,5\}, \{3,6\}$. There are $4\binom{6}{2} = 60$ such choices.

The contradiction contexts of Fig.~\ref{fig:doily-contras}(c) may be described by choosing a number $i \in \{1,2,\dots,6\}$ to remove, and then forming a ``cycle''. By forming a cycle we mean the following procedure. One first writes the remaining five elements $j_1,\dots,j_5$ of $\{1,2,\dots,6\}$ on the sides of a pentagon in any order, with two orderings being equivalent whenever one ordering can be obtained from the other by rotating or flipping the pentagon. Next, one identifies each vertex of the pentagon with the pair of sides that are incident to it. These form the $5$ contexts. An example of this is $\{1,3\},\{3,5\},\{5,4\},\{4,2\},\{2,1\}$. This is formed by first removing $6$, and then considering the ordering $1,3,5,4,2$. There are $6\binom{4}{2} = 72$ such choices. 

There is another way to view the set of all two-qubit observables that we also make use of: as a decomposition into $10$ Peres-Mermin squares. Each of these squares may be viewed as a way of partitioning the set $\{1,2,\dots,6\}$ into two sets of size three in the following way. Say we choose sets $\{1,2,4\},\{3,5,6\}$. We then form the six blocks of the grid to be the pairs $\{1,2\},\{1,4\},\{2,4\}$ from the first set along with the pairs $\{3,5\},\{3,6\},\{5,6\}$ of the second set. Together, these six contexts form a square with the pairs from the first set being viewed as the ``rows'' and the pairs from the second set being viewed as the ``columns.''
In such a way, we obtain $10$ squares, with each context appearing in exactly four squares. 

Each set of contradiction contexts must intersect one of the $10$ squares in at least three contexts. The contradiction contexts of Fig.~\ref{fig:doily-contras}(a) are exactly the $20$ choices of all rows or all columns from one of the $10$ squares. It is easy to see that the contradiction contexts of Fig.~\ref{fig:doily-contras}(b) and Fig.~\ref{fig:doily-contras}(c) both intersect some square at two rows and one column (or vice-versa).


\subsection{Memory cost of satisfying predictions (Ia) and (II) for the set of all two-qubit Pauli observables} \label{sec:doily-IaII}


An example of a 6-state solution satisfying (Ia) and (II) is the following, using the notation of \eqref{epm}:
\begin{equation} \label{eq:doily-6state}
\scalebox{0.8}{$
\begin{array}{ll}
S_1 = 
\begin{bmatrix}
         & + & + & + \\
        + & + & + & + \\
        + & (+,4) & (+,5) & (+,6) \\
        + & (+,5) & (+,3) & (+,4) \\
    \end{bmatrix},
&S_2 = 
\begin{bmatrix}
         & - & - & - \\
        - & + & + & + \\
        + & (-,3) & (-,6) & (-,5) \\
        + & (-,6) & (-,4) & (-,3) \\
    \end{bmatrix},\\[4em]   
S_3 = 
\begin{bmatrix}
         & - & (+,1) & (-,2) \\
        (-,2) & + & + & + \\
        + & - & (+,5) & (+,6) \\
        + & (-,6) & + & - \\
    \end{bmatrix},
&S_4 =
\begin{bmatrix}
         & + & (-,2) & (+,1) \\
        (+,1) & + & + & + \\
        + & + & (-,6) & (-,5) \\
        + & (+,5) & - & + \\
    \end{bmatrix},\\[4em]
S_5 = 
\begin{bmatrix}
         & + & (+,1) & (-,2) \\
        (+,1) & + & + & + \\
        + & (+,4) & + & - \\
        + & + & (-,4) & (-,3) \\
    \end{bmatrix},
&S_6 = 
\begin{bmatrix}
         & - & (-,2) & (+,1) \\
        (-,2) & + & + & + \\
        + & (-,3) & - & + \\
        + & - & (+,3) & (+,4) \\
    \end{bmatrix}.
\end{array}$}
\end{equation}

The contradiction contexts for states $S_1$ and $S_2$ are exactly the negative contexts of the set of all two-qubit observables in Fig.~\ref{fig:extendedPM-magicset}. For states $S_3$ and $S_4$, they are the contexts 
$\{\chi_{02},\chi_{10},\chi_{12}\}$, $\{\chi_{03},\chi_{20},\chi_{23}\}$, $\{\chi_{31},\chi_{13},\chi_{22}\}$. 
For the remaining two states they are the contexts 
$\{\chi_{02},\chi_{30},\chi_{32}\}$, $\{\chi_{10},\chi_{03},\chi_{13}\}$, $\{\chi_{12},\chi_{21},\chi_{33}\}$.

We observe that the Mealy machine solution for the Peres-Mermin square \eqref{eq:square-4state} has a set of simple points (i.e., the corresponding vertices are simple for each state) that are pairwise-incompatible and with consistent outputs (the diagonal with all $+$ output in each state). In the same way, the solution for the set of all two-qubit observables \eqref{eq:doily-6state} has a set of five pairwise-incompatible points that remain simple and with output $+$ in all states, in our example that is the set $\{\chi_{11}, \chi_{12}, \chi_{13}, \chi_{20}, \chi_{30}\}$.

In fact, once this condition is enforced, and the set of three disjoint contradiction contexts has been chosen, there are only two possible ways to complete the output function for any given state. This yields six possible states, for which there is a single possible transition function $\Upsilon$ (note that this assumes that nonsimple vertices map to simple vertices: a fact whose proof we have ommitted since it is somewhat lengthy and since this fact was not used in any of the arguments for optimality). In other words, the six states are completely defined by enforcing that the pairwise-incompatible points are simple and have output $1$, along with the choice of contradiction contexts. 


Now, we prove that the set of all two-qubit observables needs more memory to satisfy predictions (Ia) and (II) than the Peres-Mermin square.

\begin{proposition}
    The set of all two-qubit observables needs at least four states to satisfy predictions (Ia) and (II).
\end{proposition}

\begin{proof}
Assume towards a contradiction that there is a three-state Mealy machine with states $S_1,S_2,S_3$ that satisfies predictions (Ia) and (II). Each block $b \in B$ can form a contradiction context for at most one state (that is, for any distinct pair $i,j \in \{1,2,3\}$, $(S_i,b)$ and $(S_j,b)$ cannot both be contradiction contexts). One of the 10 squares must have, at least, three contradiction contexts $(S_1,b_1), (S_1, b_2), (S_1, b_3)$ for state $S_1$. Both $S_2$ and $S_3$ must also have at least one contradiction context $(S_2, b_4), (S_3,b_5)$ in the same square for distinct $b_4,b_5$. It is straightforward to check that either $|b_4 \cap (b_1 \cup b_2 \cup b_3)| \geq 2$ or $|b_5 \cap (b_1 \cup b_2 \cup b_3)| \geq 2$. Therefore, without loss of generality, the contradiction context $(S_2,b_4)$ has at least two simple vertices. This contradicts Proposition~\ref{prop:Db-2nonsimple}.
\end{proof}

Proving that the set of all two-qubit observables requires more than four states cannot be done by this sort of analysis and requires an examination of outputs as well. One potential avenue to exploit is the following obvious fact: for each of the (at least five) nonsimple vertices of each state, any context containing any of these nonsimple vertices cannot be simple. On the other hand, each contradiction context of a state implies at least one simple context in another state. 

In summary, the memory cost of satisfying (Ia) and (II) for the set of all two-qubit observables is at least $\log_2(4) = 2$ bits, and at most $\log_2(6)$ bits. Proving that the Mealy machine of \eqref{eq:doily-6state} is optimal would thus entail ruling out Mealy machines with four or five states.



\subsection{Memory cost of satisfying predictions (Ia), (Ib), and (II) for all two-qubit Pauli observables}


Here we take the exact same approach as in Sec.~\ref{sec:pentagram-iabii}, where we proved that no four-state Mealy machine could simulate predictions (Ia), (Ib), and (II) for Mermin's pentagram.

\begin{proposition} \label{prop:doily-1lonesink}
    No state can be in a simple sink of $\mathcal{D}_p$ for more than one point $p$.
\end{proposition}

\begin{proof}
    Assume towards a contradiction that state $S_1$ is in a simple sink in $\mathcal{D}_p$ and in $\mathcal{D}_{p'}$ for distinct points $p,p' \in P.$ Then, for each point $r$ that is adjacent to $p$ or $p'$, the vertex $(S_1,r)$ is simple. No matter the choice of $p,p'$, we obtain $11$ simple vertices. This is a contradiction since each state must have at least $5$ nonsimple vertices (Proposition~\ref{prop:doily-nonsimple}).
\end{proof}

\begin{theorem}
    There is no $5$-state Mealy machine that satisfies predictions (Ia), (Ib), and (II) for the extended PM square.
\end{theorem}

\begin{proof}
    Let $X_i$ be the set of multi-sink points $p$ for which exactly $i$ states are not in any sink of $\mathcal{D}_p$, and let $x_i = |X_i|$. 
    Each state must not be in a sink for at least $5$ multi-sink points by Proposition~\ref{prop:doily-nonsimple}. Therefore,
    \begin{equation} \label{ineq:1}
        \sum_{i=1}^{3} ix_i \geq 25.
    \end{equation}

    For each multi-sink point $p \in X_2$, there is at least one state in a simple sink of $\mathcal{D}_p$, and for each multi-sink point $p' \in X_3$, there are exactly two states in simple sinks of $\mathcal{D}_{p'}$. Therefore, 
    \begin{equation} \label{ineq:2}
        x_2 + 2x_3 \leq 5.
    \end{equation}
    In total, the set of all two-qubit observables has $15$ points, and so we also have the inequality
    \begin{equation} \label{ineq:3}
        x_0 + x_1 + x_2 + x_3 \leq 15.
    \end{equation}
    Summing inequalities~\eqref{ineq:2} and \eqref{ineq:3}, we find that 
    \begin{equation} \label{ineq:4}
        x_0 + x_1 + 2x_2 + 3x_3 \leq 20.
    \end{equation}
    Clearly, inequalities~\eqref{ineq:1} and \eqref{ineq:4} are incompatible since $x_i \geq 0$ for each $i=0,1,2,3$.
\end{proof}

\section{Conclusion and open questions}
\label{conc}


We have studied the memory cost of classically simulating two different subsets of deterministic predictions [what we have called  predictions (Ia) and (II), and predictions (Ia), (Ib), and (II)] of quantum theory. 

The memory cost of such predictions was computed for the Peres-Mermin square in \cite{Kleinmann:2011NJP}. To the best of the authors' knowledge, these calculations represent the only cases in the literature where the memory cost of deterministic predictions was computed exactly.

In this article, we have computed the exact memory cost for both of these subsets of predictions for another fundamental set: Mermin's pentagram. We found that, for predictions (Ia) and (II), the memory cost is $\log_2(4) =2$ bits, and that, for predictions (Ia), (Ib), and (II), the memory cost is $\log_2(5)  \approx 2.32$ bits.

Moreover, we improved the bounds for yet another fundamental set: the set of all 15 two-qubit Pauli observables. We showed that the memory cost of satisfying predictions (Ia) and (II) is between $\log_2(4) = 2$ and $\log_2(6) \approx 2.58$ bits. We also showed that the memory cost of satisfying predictions (Ia), (Ib), and (II) is, at least, $\log_2(6) \approx 2.58$ bits.


Several questions remain open:

While the memory cost of classically simulating all the {\em probabilistic} predictions of quantum theory is known for the Peres-Mermin square \cite{Cabello:2018PRLEPSILON} and can be computed for other observables using the tools in \cite{Cabello:2018PRLEPSILON} (e.g., for Mermin's pentagram the memory cost is $\approx 7.09$ bits), the exact memory cost of simulating all {\em deterministic} predictions of quantum theory is not known for any set producing state-independent contextuality. 

We still do not know the exact memory cost of simulating any subset of deterministic predictions for the set of all 15 two-qubit Pauli observables. What is this value for simulating predictions (Ia) and (II)? How about simulating predictions (Ia), (Ib), and (II)? For this purpose, the tools introduced in \cite{Vieira:2022Q} might be useful.
It would be interesting to obtain bounds on the classical memory cost of satisfying both subsets of quantum predictions. Perhaps as a function of the number of qubits, or some property of its underlying configuration. It would also be interesting to quantify the gap between the memory cost in these scenarios in comparison to that of simulating all deterministic predictions of quantum theory.

In each of the optimal Mealy machines (for both subsets of predictions considered), it appears that there is a maximal set of pairwise incompatible points (independent set) that are simple and each have the same output for each state. An interesting question is whether it is always possible to construct an optimal machine satisfying these conditions.

Given that the memory cost of classically simulating the deterministic predictions of quantum theory for the sets of $n$-qubit Pauli observables grows as $O(n^2)$ \cite{Harrysson:2016}, it would be interesting to explore the memory cost for other sets (even dropping the requirement of state independence) whose classical simulation is, in principle, much more difficult \cite{AmaralPRA2015}. 


%

\end{document}